\documentclass{llncs}
\usepackage[utf8]{inputenc}
\usepackage{amsmath}
\usepackage{amssymb}
\usepackage{enumerate}
\usepackage{mathrsfs}
\usepackage{float}
\usepackage{pgf,tikz,pgfplots}
\pgfplotsset{compat=1.14}
\usetikzlibrary{arrows}
\definecolor{xdxdff}{rgb}{0.49019607843137253,0.49019607843137253,1}
\definecolor{ududff}{rgb}{0.30196078431372547,0.30196078431372547,1}
\definecolor{ffvvqq}{rgb}{1,0.3333333333333333,0}
\definecolor{qqaydv}{rgb}{0,0.6588235294117647,0.8352941176470589}
\definecolor{ccqqqq}{rgb}{0.8,0,0}
\newcommand{\Z}{\mathbb{Z}}
\newcommand{\Q}{\mathbb{Q}}
\newcommand{\R}{\mathbb{R}}
\renewcommand{\deg}{\mathtt{deg}}
\newcommand{\coeff}{\mathtt{coeff}}
\newcommand{\var}{\mathtt{var}}
\newcommand{\lc}{\mathtt{lc}}
\newcommand{\lt}{\mathtt{lt}}
\newcommand{\res}{\mathtt{res}}
\newcommand{\disc}{\mathtt{disc}}
\newcommand{\poly}{\mathtt{poly}}
\newcommand{\Root}{\mathtt{Root}}
\newcommand{\order}{\mathtt{order}}
\newcommand{\level}{\mathtt{level}}
\newcommand{\sample}{\mathtt{sample}}
\newcommand{\feasible}{\mathtt{feasible}}
\newcommand{\Value}{\mathtt{value}}
\newcommand{\proj}{\mathtt{Proj}}
\newcommand{\samplepoly}{\mathtt{s\_poly}}
\newcommand{\sampleinterval}{\mathtt{s\_interval}}
\newcommand{\samplecoeff}{\mathtt{s\_coeff}}
\newcommand{\samplecell}{\mathtt{s\_cell}}
\newcommand{\conflictcore}{\mathtt{core}}
\newcommand{\explain}{\mathtt{explain}}
\newcommand{\resolve}{\mathtt{resolve}}
\newcommand{\factor}{\mathtt{factor}}

\title{Solving Satisfiability of Polynomial Formulas By Sample-Cell Projection  }
\author{Haokun Li  \and Bican Xia}
\institute{School of Mathematical Sciences, Peking University, Beijing, China\\
\email{ker@pku.edu.cn}, \email{xbc@math.pku.edu.cn}}

\begin{document}

\maketitle
\begin{abstract}
    A new algorithm for deciding the satisfiability of polynomial formulas over the reals is proposed. The key point of the algorithm is a new projection operator, called sample-cell projection operator, custom-made for Conflict-Driven Clause Learning (CDCL)-style search. Although the new operator is also a CAD (Cylindrical Algebraic Decomposition)-like projection operator which computes the cell (not necessarily cylindrical) containing a given sample such that each polynomial from the problem is sign-invariant on the cell, it is of singly exponential time complexity. The sample-cell projection operator can efficiently guide CDCL-style search away from conflicting states. Experiments show the effectiveness of the new algorithm.
    
    %We propose a new algorithm for deciding satisfiability of polynomial formulas   over the Real. The major improvement is a projection operator custom-made for Conflict-Driven Clause Learning (CDCL)-style search , called sample-cell projection operator.   The sample-cell projection operator can efficiently guide  CDCL-style search away from conflicting states. The new operator is also cylindrical algebraic decomposition (CAD)-like  projection operator for computing the cell which contains a given sample and is sign-invariant on all  polynomials in problem, but only of singly exponential time complexity. 
    
%      The major improvement is a new improved cylindrical algebraic decomposition (CAD) projection operator for computing the cell which contains a given sample and is sign-invariant on all  polynomials in problem, called sample-cell projection operator. The new operator is of singly exponential time complexity. The sample-cell projection operator can efficiently guide Conflict-Driven Clause Learning (CDCL) search away from conflicting states.
    % We propose a new algorithm for deciding satisfiability of polynomial constraints over the Real.
    % We present a new improved projection for cylindrical algebraic decomposition (CAD), which can  efficiently guide Conflict-Driven Clause Learning (CDCL) search away from conflicting states, called sample-cell projection operator. The new operator is of singly exponential time complexity, but only can work with a sample that can be get from CDCL. 
\end{abstract}
\section{Introduction}
%    Developments in artificial intelligence, automated theorem proving and formal methods have greatly boosted from the recent advances in boolean satisfiability problem (SAT). Furthermore, however, 
The research on SMT (Satisfiability Modulo Theories) \cite{deMoura+Dutertre+Shankar:cav2007,DBLP:journals/cacm/MouraB11,DBLP:series/faia/2009-185} in recent years %gives many new algorithms for solving problems in various theories such as EUF, Linear Arithmetic over the integers and the reals, etc. [???] and 
brings us many popular solvers such as Z3 \cite{DBLP:conf/tacas/MouraB08}, CVC4 \cite{BCD+11}, Yices \cite{Dutertre:cav2014}, MathSAT5 \cite{mathsat5}, etc.
Nevertheless, in theory and practice, it is important to design efficient SMT algorithms and develop tools (or improve existing ones) for many other theories, {\it e.g.} string \cite{DBLP:conf/cav/LiangRTBD14}, linear arithmetic \cite{DBLP:conf/cav/DutertreM06,DBLP:conf/cade/JovanovicM12} and non-linear arithmetic \cite{DBLP:journals/corr/abs-1905-09227,DBLP:conf/smt/KorovinKS14} over the reals. A straightforward idea is to integrate  Conflict-Driven Clause Learning (CDCL)-style search with theory solvers \cite{DBLP:series/faia/2009-185}.
For example, integrating CDCL-style search with a theory solver for determining whether a basic semialgebraic set is empty can solve satisfiability in the theory of non-linear arithmetic over the reals. 

It is well-known that the problem whether a basic semialgebraic set is empty is decidable due to Tarski’s decision procedure \cite{10.1007/978-3-7091-9459-1_3}. Tarski's algorithm cannot be a theory solver in practice because of its very high complexity.
%Traditionally, Methods such as Gr\"obner basis \cite{DBLP:journals/cca/Buchberger76}, Wu's method of characteristic set \cite{Xia_Yang_2016} and sum of squares (SOS) optimization \cite{POWERS199899}\cite{DBLP:conf/issac/WangLX19} can solve many of the related problems, and cylindrical algebraic decomposition (CAD) algorithm is most general and effective \cite{DBLP:journals/cca/Collins76}.
Cylindrical algebraic decomposition (CAD) algorithm \cite{DBLP:journals/cca/Collins76} is a widely used theory solver in practice though it is of doubly exponential time complexity.  
   % Alfred Tarski groundbreakings proposed  Tarski’s decision procedure \cite{10.1007/978-3-7091-9459-1_3}, and based on that he proved that the real closed fields theory has quantifiers elimination. Further, Collins \cite{DBLP:journals/cca/Collins76} propose the first relatively effective method by cylindrical algebraic decomposition (CAD).
The idea of CAD algorithm is to decompose $\R^n$ into cells such that each polynomial from the problem is sign-invariant in every cell. A key concept in CAD algorithm is the projection operator. %McCallum  \cite{DBLP:journals/jsc/McCallum88}\cite{10.1007/978-3-7091-9459-1_12} and Brown \cite{brown_improved_2001} had given a smaller projection. %Xia, Dai and Han presented open week cad to improve the projection operator for computing sample for all open connected components for a given polynomial \cite{Han_Dai_Xia_2014}\cite{Dai_Han_Hong_Xia_2015}\cite{Xia_Yang_2016}.
Although many improved projection operators have been proposed \cite{DBLP:conf/issac/Hong90,DBLP:journals/jsc/McCallum88,10.1007/978-3-7091-9459-1_12,brown_improved_2001,Han_Dai_Xia_2014,Dai_Han_Hong_Xia_2015,Xia_Yang_2016}, the CAD method is still of doubly exponential time complexity. 
The main reason is that in order to carry enough information, projection of variables causes the number of polynomials grows rapidly. 
%Inevitably, number of polynomials will grow doubly exponentially in the number of variables, to be a main hurdle to CAD Application. 
So the cost of simply using CAD as a theory solver is unacceptable.

%    Jovanovic and de Moura \cite{DBLP:conf/cade/JovanovicM12} try to alleviate this problem for the existential theory by combining Conflict-Driven Clause Learning (CDCL) and CAD. 
%    %CDCL is an important method in SAT and SMT \cite{DBLP:series/faia/2009-185}. 
%    They used CDCL-style  search on each variable. For a variable, CDCL-style search used real-root isolation to check  consistency and used CAD to explain the conflict. This method expands the scope of the problem that can be solved, especially if the problem exists a model.

   %But for some conflict problem, The above method is less efficient than the simple CAD. In this paper, we propose a new algorithm. The major improvement is that we present a new CAD projection operator. The new operator is custom-made for CDCL-style search, and is of singly exponential time complexity. It only  computes the cell which contains a given sample and is sign-invariant on all polynomials in problem. And this cell is exactly what conflict explanation needs.
   
Jovanovic and de Moura \cite{DBLP:conf/cade/JovanovicM12}  eased the burden of using CAD as a theory solver by modifying the CDCL-style search framework. They changed the sequence of search states by adding variable assignments to the sequence. The benefit of this is that they can use real-root isolation, which is of polynomial time complexity, to check consistency of literals for there will be only one unassigned variable in the literals of the current state. When a conflict of literals is detected, they explain the conflict by applying CAD to a polynomial set called conflicting core to find the cell where the sample of assignments belongs. 
But even using CAD only when explaining conflicts is a huge computational cost, as CAD is of doubly exponential time complexity.  Furthermore, CAD will produce all cells in $\R^n$  other than the only one we need, making computation waste. 

In this paper, we propose a new custom-made CAD-like projection operator, called sample-cell projection operator. It only processes the cell containing a given sample, which is exactly what conflict explanation needs.
The idea of our operator is trying to project polynomials related to the target cell and ignore irrelevant polynomials.
We integrate our sample-cell projection operator with Jovanovic's improved CDCL-style search framework. The new operator can efficiently guide CDCL-style search away from conflicting states.
It is proved that the new algorithm is of singly exponential time complexity. %, saving much computation cost.
We have implemented a prototype solver LiMbS which is base on Mathematica 12. %to demonstrate the effectiveness of our algorithm.  
Experiments show the effectiveness of the new algorithm.
   
The rest of this paper is structured as follows: Section \ref{sec:pre} introduces  the  background  knowledge and notation. 
Section \ref{sec:sample} defines sample-cell projection and presents the details of our approach. Section \ref{sec:cdcl} describes the CDCL-style search framework which we adopt.
We evaluate our approach on many well-known examples and analyze its performance in Section \ref{sec:exp}. The paper is concluded in Section \ref{sec:conclusion}.

%\section{Preliminaries}\label{sec:pre}
\section{Notation}\label{sec:pre}
Let $\R$ denote the field of real numbers, $\Z$ denote the ring of integers and $\Q$ denote the field of rational numbers. %Let $\Z[\bar{x}]$ be the ring of multi-variate polynomials in variables $\bar{x}$ with integer coefficients. 
Unless stated otherwise, we assume that all polynomials in this paper are in $\Z[\bar{x}]$, the ring of multivariate polynomials in variables $\bar{x}$ with integer coefficients.
 
For a polynomial $f\in \Z[\bar{y},x]$:
$$f(\bar{y},x)=a_mx^m+a_{m-1}x^{m-1}+\ldots+a_1x+a_0$$
where $a_m\neq0$ and $a_i\in \Z[\bar{y}]$ for $i=0,...,m$, the {\em degree} of
$f$ with respect to (w.r.t.) $x$ is $m$, denoted by $\deg(f,x)$. The {\em leading coefficient} of $f$ w.r.t. $x$ is $a_m$, denoted by $\lc(f,x)$ and
the {\em leading term} of $f$ w.r.t. $x$ is $a_mx^{m}$, denoted by $\lt(f,x)$.
Let 
\[\coeff(f, x)=\{a_i|0\leq i \leq m \land a_i\neq 0\}\] denote the {\em set of coefficients} of $f$ w.r.t. $x$ and $\var(f)=\{\bar{y},x\}$
denote the variables appearing in $f$. %Call the polynomial univariate if $|\var(f)|=1$, multivariate if  $|\var(f)|>1$ and constant if $|\var(f)|=0$.

Suppose $g\in \Z[\bar{y},x]$:
$$g(\bar{y},x)=b_nx^{n}+b_{n-1}x^{n-1}+\ldots+b_1x+b_0$$
where $b_n\neq 0$ and $b_i\in \Z[\bar{y}]$ for $i=0,...,n$ .
Let $\res(f,g,x)$ denote the Sylvester {\em resultant} of $f$ and $g$ w.r.t. $x$, {\it i.e.} the determinant of the following matrix
\[ 
\left(\begin{array}{cccccccc} 
    a_m    & a_{m-1}& a_{m-2}& \ldots& a_0   & 0    & \ldots  &      0 \\ 
    0      & a_m    &a_{m-1} & \ldots& a_1   & a_0  & \ldots  &      0 \\ 
    \vdots & \vdots &\ddots  & \ddots&\ddots &\ddots& \ddots  & \vdots \\
    0      & 0      &\ldots  & a_m   &a_{m-1}&\ldots& \ldots  &    a_0 \\ 
    b_n    & b_{n-1}& b_{n-2}& \ldots& b_0   & 0    & \ldots  &      0 \\ 
    0      & b_n    &b_{n-1} & \ldots& b_1   & b_0  & \ldots  &      0 \\ 
    \vdots & \vdots &\ddots  & \ddots&\ddots &\ddots& \ddots  & \vdots \\
    0      & 0      &\ldots  & b_n   &b_{n-1}&\ldots& \ldots  &    b_0 \\ 
\end{array}\right)
\]
which has $n$ rows of $a_i$ and $m$ rows of $b_j$. The discriminant of $f$ w.r.t. $x$ is
$$\disc(f,x)=(-1)^\frac{m(m-1)}{2}\res(f,f',x).$$

%We call $a\in\R$ a root of a univariate polynomial $f$ iff $f(a)=0$. Correspondingly, $a$ is called algebraic iff there exists a univariate polynomial $f$ such that $f(a)=0$, and we denote the field of real algebraic numbers by $\R_{alg}$.

An {\em atomic polynomial constraint} is $f\triangleright0$ where $f$ is a polynomial and $\triangleright\in \{\geq,>,=\}$. A {\em polynomial literal} (simply {\em literal}) is an atomic polynomial constraint or its negation. For a literal $l$, $\poly(l)$ denotes the polynomial in $l$ and $\var(l)=\var(\poly(l))$.  A {\em polynomial clause} is a disjunction $l_1\lor\cdots\lor l_s$ of literals. Sometimes, we write a clause as $\lnot(\bigwedge_i l_i)\lor \bigvee_j l_j$. A {\em polynomial formula} is a conjunction of clauses. % $c_1\land\ldots\land c_n$.
%In fact, we will working with  extended polynomial constraints from  extended Tarski language. 
An {\em extended polynomial constraint} $l$ is $x\triangleright \Root(f,k)$ where $\triangleright\in \{\geq,>,=\}$, $f\in\Z[\bar{y},u]$ with $x\not\in \var(f)$ and $k (0\leq k\leq \deg(f,u))$ is a given integer. Notice the variable $u$ is an exclusive free variable that cannot be used outside the $\Root$ object.% and cannot be in any assignment. 

For a formula $\phi$, $\phi[a/x]$ denote the resulting formula via substituting $a$ for $x$ in $\phi$.
For variables $\bar{x}=(x_1,\ldots,x_r)$ and $\bar{a}=(a_1,\ldots,a_r)\in\R^r$, a mapping $\alpha$ which maps $x_i$ to $a_i$ for $i=1,...,r$ is called a {\em variable assignment} of $\bar{x}$ and $\bar{a}$ is called a {\em sample} of $\alpha$ or a {\em sample} of $\bar{x}$ in $\R^r$. 
%For polynomials $f\in\Z[\bar{x}]$ and $g\in\Z[\bar{x},y]$, $\alpha(f)\in \R$ is  the value  of $f$ under $\alpha$ and   $\alpha(g)\in \R[y]$ is the polynomial by $g$ under $\alpha$. We say $f$ or $g$ vanishes under $\alpha$ iff $\alpha(f)=0$ or  $\alpha(g)=0$.
We denote $\phi[a_1/x_1,\ldots,a_r/x_r]$ by $\alpha(\phi)$. If $\alpha(\phi)=0$, we say $\phi$ vanishes under $\alpha$  or vanishes under $\bar{a}$. 
Suppose an extended polynomial constraint $l$ is of the form $x\triangleright \Root(f,k)$ and $\alpha$ is a variable assignment of $(\bar{y},x)$. If $\beta_k$ is the $k$th real root of $\alpha(f)$, $\alpha(l)$ is defined to be $\alpha(x)\triangleright \beta_k$. If $\alpha(f)$ has less than $k$ real roots,  $\alpha(l)$ is defined to be {\tt False}. %By extended polynomial constraints, we can easily represent cells.

\section{Sample-Cell Projection}\label{sec:sample}
%As the foreword goes, we propose Sample-Cell Projection to find the cell where the sample belongs efficiently. We're going to start with some of the concepts we need.
In this section, we first introduce some well-known concepts and results concerning CAD and then define the so-called sample-cell projection operator.

%In modern cylindrical algebraic decomposition  systems, we introduce order-invariant to substitute sign-invariant. 
Let $f$ be an  analytic function  defined in some open set $U$ of $K^n$ where $K$ is a field. For a point $p\in U$, if $f$ or some partial derivative (pure and mixed) of $f$ of some order does not vanish at $p$, then we say that $f$ has {\em order} $r$ where $r$ is the least non-negative integer such that some partial derivative of total order $r$ does not vanish at $p$. Otherwise, we say $f$ has infinite order at $p$. The order of $f$ at $p$ is denoted by $\order_pf$.  We say $f$ is {\em order-invariant} in a subset $S\subset U$ if $\order_{p_1}f=\order_{p_2}f$ for any  $p_1, p_2\in S$.  Obviously, if $K=\R$ and the analytic function $f$ is order-invariant in $S$, then $f$ is sign-invariant in $S$.

An $r$-variable polynomial $f(\bar{x},x_r)$ where $\bar{x}=(x_1,\ldots,x_{r-1})$ is said to be {\em analytic delineable} on a connected  $s$-dimensional submanifold $S\subset \R^{r-1}$ if 
\begin{enumerate}
    \item The number $k$ of different real roots of $f(a,x_r)$ is invariant for any point $a\in S$. And the trace of the real roots  are the graphs of some pairwise disjoint analytic functions $\theta_1<\ldots<\theta_k$ from $S$ into $\R$ ({\it i.e.} the order of real roots of $f(a,x_r)$ is invariant for all point $a\in S$);
    \item There exist positive integers $m_1,\ldots,m_k$ such that for every point $a\in S$, the multiplicity of the real root $\theta_i(a)$ of $f(a,x_r)$ is $m_i$ for $i=1,...,k$.
\end{enumerate}
Especially, if $f$ has no zeros in $S\times \R$, then $f$ is delineable on $S$ with $k=0$. The analytic functions $\theta_i$'s are called the {\em real root functions} of $f$  on $S$, the graphs of the $\theta_i$'s are called the {\em $f$-sections} over $S$, and the connected regions between two consecutive $f$-sections (for convenience, let $\theta_0=-\infty$ and $\theta_{k+1}=+\infty$) are called {\em $f$-sectors} over $S$. Each $f$-section over $S$ is a connected  $s$-dimensional submanifold in $\R^r$ and each  $f$-sector over $S$ is a connected  $(s+1)$-dimensional submanifold in $\R^r$. 

\begin{theorem}[\cite{10.1007/978-3-7091-9459-1_12}, Theorem 2]\label{thm:mc}
Let $r\geq 2$ and $f(\bar{x},x_r)$ be a polynomial in $\R[\bar{x},x_r]$ of positive degree where $\bar{x}=(x_1,...,x_{r-1})$.  Let $S$ be a connect submanifold of $\R^{r-1}$ where $f$ is degree-invariant and does not vanish identically. Suppose that $\disc(f,x_r)$ is a nonzero polynomial and is order-invariant in $S$. Then $f$ is analytic delineable on $S$ and is order-invariant in each $f$-section over $S$.
\end{theorem}

%Suppose $S\in \R^{(n-1)}$ is  a connected submanifold, $F=\{f_1(\bar{x},x_n),\ldots,f_r(\bar{x},x_n)\}$ is a polynomial set in $\Z[\bar{x},x_n]$ where $\bar{x}=\{x_1,\ldots,x_{n-1}\}$ and each $f_i$ is delineable on $S$. For a sample $x_0=(x_{0,1},\ldots,x_{0,n})$ of $\{\bar{x},x_n\}$ in $\R^n$ where $(x_{0,1},\ldots,x_{0,n-1})\in S$, 
Suppose $a=(\bar{a},a_n)=(a_1,\ldots,a_n)$ is a sample  of $(\bar{x},x_n)$ in $\R^n$ and $F=\{f_1(\bar{x},x_n)$ $,\ldots , f_r(\bar{x},x_n)\}$ is a polynomial set in $\Z[\bar{x},x_n]$ where $\bar{x}=(x_1,\ldots,x_{n-1})$. Consider the real roots of polynomials in $\{f_1(\bar{a},x_n),\ldots,f_r(\bar{a},x_n)\}\setminus \{0\}$. Denote the $k$th real root of $f_i(\bar{a},x_n)$ by $\theta_{i,k}$.  
%If exists $\theta_{i,k}=a_n$, then the sample polynomials set of $a$ in $F$ is $\{f_i\}$ written as $\samplepoly(F,a)$  and the sample interval  of $a$ in $F$ is $[\Root(f_i,k),\Root(f_i,k)]$ written as $\sampleinterval(F,a)$. On the hands, if 
We define two concepts: the {\em sample polynomials set} of $a$ in $F$ (denoted by $\samplepoly(F,x_n,a)$) and the {\em sample interval} of $a$ in $F$ (denoted by $\sampleinterval(F,x_n,a)$) as follows.

If there exists $\theta_{i,k}$ such that $\theta_{i,k}=a_n$ then
\begin{align*}
    \samplepoly(F,x_n,a)&=\{f_i\},\\
    \sampleinterval(F,x_n,a)&=(x_n=\Root(f_i(\bar{x},u),k));
\end{align*}

If there exist two consecutive real roots $\theta_{i_1,k_1}$ and $\theta_{i_2,k_2}$ such that $\theta_{i_1,k_1}<a_n<\theta_{i_2,k_2}$ %and any $\theta_{i',k'}$ such that  $\theta_{i',k'}\leq \theta_{i_1,k_1}\lor \theta_{i',k'}\geq \theta_{i_2,k_2}$  
then
\begin{align*}
    \samplepoly(F,x_n,a)&=\{f_{i_1},f_{i_2}\},\\
    \sampleinterval(F,x_n,a)&=\Root(f_{i_1}(\bar{x},u),k_1)<x_n<\Root(f_{i_2}(\bar{x},u),k_2);
\end{align*}

If there exists $\theta_{i',k'}$ such that $a_n>\theta_{i',k'}$ and for all $\theta_{i,k}$   $\theta_{i',k'}\geq\theta_{i,k}$ then
\begin{align*}
    \samplepoly(F,x_n,a)&=\{f_{i'}\},\\
    \sampleinterval(F,x_n,a)&=x_n>\Root(f_{i'}(\bar{x},u),k');
\end{align*}

If there exists $\theta_{i',k'}$ such that $a_n<\theta_{i',k'}$ and for all $\theta_{i,k}$  $\theta_{i',k'}\leq\theta_{i,k}$then
\begin{align*}
    \samplepoly(F,x_n,a)&=\{f_{i'}\},\\
    \sampleinterval(F,x_n,a)&=x_n<\Root(f_{i'}(\bar{x},u),k').
\end{align*}
Specially, if every polynomial in $\{f_1(\bar{a},x_n),\ldots,f_r(\bar{a},x_n)\}\setminus \{0\}$ does not have any real roots, define %$\samplepoly(F,x_n,(\bar{a},a_n))=\emptyset$ and $\sampleinterval(F,x_n,(\bar{a},a_n))=\text{{\tt True}}$.
\begin{align*}
    \samplepoly(F,x_n,a)&=\emptyset,\\
    \sampleinterval(F,x_n,a)&=\text{{\tt True}}.
\end{align*}
\begin{example} \label{ex:s}
Let $F=\{f_1,f_2,f_3\}$ where $f_1=y+0.5x-10$, $f_2=y+0.01(x-9)^2-7$, $f_3=y-0.03x^2-1$ and $A=(4,9),B=(4,6.75),C=(4,4),D=(4,1)$. We have (see Figure \ref{fig:1})
\[
\begin{array}{ll}
  \samplepoly(F,y,A)=\{f_1\}, &\sampleinterval(F,y,A)=y>\Root(f_1(x,u),1),\\
  \samplepoly(F,y,B)=\{f_2\},  &\sampleinterval(F,y,B)=y=\Root(f_2(x,u),1),\\
  \samplepoly(F,y,C)=\{f_2,f_3\},&\sampleinterval(F,y,C)=\begin{array}{rl}
      & y>\Root(f_3(x,u),1) \\
    \wedge &  y<\Root(f_2(x,u),1)
  \end{array},\\
 % y>\Root(f_3(x,\#),1)\land y<\Root(f_2(x,\#),1)\\
  \samplepoly(F,y,D)=\{f_3\},&\sampleinterval(F,y,D)=y<\Root(f_3(x,u),1).\\
\end{array}
\]

%Let $F=\{f,g,k\}$,$f:y+0.01(x-9)^2-7$,$g:y-0.03x^2-1$,$k:y+0.5x-10$ and $A=(4,9),B=(4,6.75),C=(4,4),D=(4,1)$. Then $\samplepoly(F,y,A)=\{k\}$, $\samplepoly(F,y,B)=\{f\}$, $\samplepoly(F,y,C)=\{f,g\}$, $\samplepoly(F,y,D)=\{g\}$
\begin{figure}[H]
    \centering
\begin{tikzpicture}[scale = 0.5,line cap=round,line join=round,>=triangle 45,x=1cm,y=1cm]
\begin{axis}[
x=1cm,y=1cm,
axis lines=middle,
ymajorgrids=true,
xmajorgrids=true,
xmin=-0.5865830202854977,
xmax=12.486294515401967,
ymin=-0.6301322314049569,
ymax=10.113669421487607,
xtick={0,1,...,12},
ytick={0,1,...,10},]
\clip(-0.5865830202854977,-0.6301322314049569) rectangle (12.486294515401967,10.113669421487607);
\draw[line width=2pt,color=ccqqqq,smooth,samples=100,domain=-0.5865830202854977:12.486294515401967] plot(\x,{0-0.01*((\x)-9)^(2)+7});
\draw[line width=2pt,color=qqaydv,smooth,samples=100,domain=-0.5865830202854977:12.486294515401967] plot(\x,{0-0.5*(\x)+10});
\draw[line width=2pt,color=ffvvqq,smooth,samples=100,domain=-0.5865830202854977:12.486294515401967] plot(\x,{0.03*(\x)^(2)+1});
\begin{scriptsize}
\draw (5,-0.4) node[scale=3] {$x$} ;
\draw (-0.4,5) node[scale=3] {$y$} ;
\draw[color=ccqqqq] (1.4062674680691189,5.97677986476334) node[scale=3] {$f_2$} ;
\draw[color=qqaydv] (1.0745740045078913,9.093459053343354) node[scale=3] {$f_1$};
\draw[color=ffvvqq] (1.0062674680691189,1.5978918106686724) node[scale=3] {$f_3$};
\draw [fill=ududff] (4,9) circle (2.5pt);
\draw[color=ududff] (4.161726521412478,9.45251239669422) node[scale=2] {$A$};
\draw [fill=xdxdff] (4.002201149761497,6.7502200665255465) circle (2.5pt);
\draw[color=xdxdff] (4.161726521412478,7.198567993989485) node[scale=2] {$B$};
\draw [fill=ududff] (4,4) circle (2.5pt);
\draw[color=ududff] (4.161726521412478,4.44875582268971) node[scale=2] {$C$};
\draw [fill=ududff] (4,1) circle (2.5pt);
\draw[color=ududff] (4.161726521412478,1.4434966190833982) node[scale=2] {$D$};
\end{scriptsize}
\end{axis}
\end{tikzpicture}
    \caption{Example \ref{ex:s}}
    \label{fig:1}
\end{figure}
\end{example}
Additionally, for a polynomial
 $$h=c_mx_n^{d_m}+c_{m-1}x_n^{d_{m-1}}+\ldots+c_{0}x_n^{d_0}$$
where $d_m>d_{m-1}>\cdots>d_0$, $c_i\in \R[\bar{x}]$ and $c_i\neq0$ for $i=0,...,m$. If there exists $j\ge 0$ such that $c_j(\bar{a})\neq 0$ and $c_i(\bar{a})=0$ for any $i>j$, then the {\em sample coefficients} of $h$ at $(\bar{a},a_n)$ is defined to be $\{c_m,c_{m-1},\ldots,c_j\}$, denoted by  $\samplecoeff(h,x_n,(\bar{a},a_n))$. Otherwise  $\samplecoeff(h,x_n,(\bar{a},a_n))=\{c_m,\ldots,c_0\}$.
\begin{definition}\label{def:projsc}
Suppose $\bar{a}$ is a sample of $\bar{x}$ in $\R^n$ and $F=\{f_1,\ldots , f_r\}$ is a polynomial set in $\Z[\bar{x}]$  where $\bar{x}=(x_1,\ldots,x_n)$. The {\em sample-cell projection} of $F$ on $x_n$ at $\bar{a}$ is 
\[
\begin{split}
    \proj_{sc}(F,x_n,\bar{a})= \bigcup_{f\in F}&\samplecoeff(f,x_n,\bar{a})\cup\\
                        \bigcup_{f\in F}&\{\disc(f,x_n)\}\cup\\
                        \bigcup_{\begin{subarray}{c}f\in F,g\in\\ \samplepoly(F,x_n,\bar{a}),\\f\neq g\end{subarray}}&\{\res(f,g,x_n)\}
\end{split}
\]
\end{definition}
%\paragraph{REMARKS} 
\begin{remark}\quad
\begin{itemize}
  \item If $f\in F$ and $x_n\not\in \var(f)$, $f$ is obviously an element of $\proj_{sc}(F, x_n, \bar{a})$. 
  %For $f\in F$ such that $x_n\not\in \var(f)$. Just add $f$ to the result, without doing anything. 
  \item  Computing $\proj_{sc}(F,x_n,\bar{a})$ will produce  $O(rn+3r)$ elements, % and run in polynomial time complexity of $n$ and $r$, 
  so the time complexity of projecting all the variables by recursively using $\proj_{sc}$ is $O((n+3)^nr)$.
\end{itemize}
\end{remark}

Now we prove the property of the new projection operator. A set of polynomials in $\Z[\bar{x}]$ is said to be a {\em squarefree basis} if the elements of the set have positive degrees, and are primitive, squarefree and pairwise relatively prime. For a connected submanifold $S$ of $\R^{n-1}$,  we denote by $S\times \sampleinterval(F,x_n,\bar{a})$
$$
\left\{(\alpha_1,\ldots,\alpha_n)\in \R^n~|~~\begin{array}{l}
     (\alpha_1,\ldots,\alpha_{n-1})\in S  \\
     \wedge~ \sampleinterval(F,x_n,\bar{a})[\alpha_1/x_1,\ldots,\alpha_n/x_n]
\end{array} \right\}.
$$
\begin{theorem}\label{th:sc}
%sample-cell Projection operator theorem
Let $F$ be a finite squarefree basis in $\Z[\bar{x}]$ where $\bar{x}=(x_1,\ldots,x_n)$ and $n\geq 2$. Let $\bar{a}=(a_1,\ldots,a_n)$ be a sample of $\bar{x}$ in $\R^n$   and $S$ be a connected submanifold of $\R^{n-1}$ such that $(a_1,\ldots,a_{n-1})\in S$. Suppose that
each element of $\proj_{sc}(F,x_n,\bar{a})$  is order-invariant in $S$. Then each element in $F$ either
vanishes identically on $S$ or  is analytic delineable on $S$, each section over $S$ of the element of $F$ which do not vanish identically on $S$ is either equal to or disjoint with $S\times \sampleinterval(F,x_n,\bar{a})$, and each element of $F$  either vanishes identically on $S$ or is order-invariant in $S\times \sampleinterval(F,x_n,\bar{a})$.
\end{theorem}
\begin{proof}
For any $f\in F$, if $f$ vanishes identically on $S$, there is nothing to prove. So  we may assume that any element in $F$ does not vanish identically on $S$.

For any $f\in F$ such that $f\not\in\samplepoly(F,x_n,\bar{a})$, %and $f$ does not vanish identically on $S$, 
let $f'=f\cdot~\prod_{g\in \samplepoly(F,x_n,\bar{a})} g$. Notice that $f'$ is degree-invariant on $S$ (each element of $\samplecoeff(f,x_n,\bar{a})$ is order-invariant, hence sign-invariant in $S$). And we have
\[
\begin{split}
\disc(f',x_n)=\disc(f,x_n)\cdot\prod_{g\in \samplepoly(F,x_n,\bar{a})} \disc(g,x_n)\cdot\\
\prod_{g\in \samplepoly(F,x_n,\bar{a})}\res(f,g,x_n)\cdot\\
\prod_{\begin{subarray}{c}g_1\in \samplepoly(F,x_n,\bar{a}),\\ g_2\in \samplepoly(F,x_n,\bar{a}),\\g_1\neq g_2\end{subarray}}\res(g_1,g_2,x_n).
\end{split}
\]
It follows from this equality that $\disc(f',x_n)\neq 0$ (because $f_i$'s are squarefree and pairwise relatively prime).
Obviously, each factor of $\disc(f',x_n)$ is a factor of $\proj_{sc}(F,x_n,\bar{a})$,
so $\disc(f',x_n)$ is order-invariant in $S$. 
By Theorem \ref{thm:mc}, $f'$ is analytic delineable on $S$ and is order-invariant in each $f'$-section over $S$. 
So $f$ and $g\in \samplepoly(F,x_n,\bar{a})$ %are delineable on $S$, and 
are  order-invariant in each $f'$-section over $S$.  It follows that the sections over $S$ of $f$ and $g$ are pairwise disjoint.
%\textcolor{red}{And because an $f'$-section is an $f$-section  or an  $\samplepoly(F,x_n,\bar{a})$-section and a root of $f'$ is a root of $f$ or a root of $\samplepoly(F,x_n,\bar{a})$.}
 %, every $f$-section is either equal to or disjoint with $S\times \sampleinterval(F,x_n,\bar{a})$. %which isn't same as $S\times \sampleinterval(F,x_n,\bar{a})$  is disjoint with $S\times \sampleinterval(F,x_n,\bar{a})$.
Therefore, $f$ and $g\in \samplepoly(F,x_n,\bar{a})$  are analytic delineable on $S$, every section of them is either equal to or disjoint with $S\times \sampleinterval(F,x_n,\bar{a})$, and $f$ and $g$ are order-invariant in $S\times \sampleinterval(F,x_n,\bar{a})$.
%Therefore, we have the conclusion that $f$ is analytic delineable on $S$, %each $f$-section which isn't same as $S\times \sampleinterval(F,x_n,\bar{a})$  is disjoint with $S\times \sampleinterval(F,x_n,\bar{a})$, 
%every $f$-section is either equal to or disjoint with $S\times \sampleinterval(F,x_n,\bar{a})$, and $f$ is order-invariant in $S\times \sampleinterval(F,x_n,\bar{a})$.
\hfill $\blacksquare$ 
\end{proof}
%\paragraph{REMARKS}
\begin{remark}\label{rm:sc}
Notice that when $f$ vanishes identically on $S$, $f$ isn't always  order-invariant in $S\times \sampleinterval(F,x_n,\bar{a})$.  %But this rarely happens. 
This is avoidable by changing the ordering of variables and is negligible  when the satisfiability set of formulas is full-dimensional. We find a way to handle this rare case: either to determine whether the coefficients of $f$ have finitely many common zeros,  
or to enlarge $F$ by adding partial derivatives of $f$ whose order is less than $\order(f)$ and  one non-zero partial derivative whose order is exactly equal to $\order(f)$.
\end{remark}

When integrating the new projection operator with the CDCL-type search (see Section \ref{sec:cdcl}), we need a traditional CAD projection operator \cite{DBLP:journals/jsc/McCallum88,10.1007/978-3-7091-9459-1_12}.
%In practice, the dimension of our polynomials  will be larger than the sample by $1$. So we first perform a traditional CAD projection on the first variable to decrease the dimension of polynomial by 1, so that the dimension of sample is same as polynomials. For the rest variables we use sample-cell projection. Here we choose  McCallum's projection operator  (\cite{DBLP:journals/jsc/McCallum88} \cite{10.1007/978-3-7091-9459-1_12}) as CAD projection.
\begin{definition}[\cite{DBLP:journals/jsc/McCallum88}]
Suppose $F=\{f_1,\ldots , f_r\}$ is a polynomial set in $\Z[\bar{x}]$  where $\bar{x}=(x_1,\ldots,x_n)$. The McCallum projection of $F$ on $x_n$ is
$$\proj_{mc}(F)=\bigcup_{f\in F}\{\coeff(f),\disc(f,x_n)\}\cup\bigcup_{\begin{subarray}{c}f\in F,g\in F,\\f\neq g\end{subarray}} \res(f,g,x_n)$$
\end{definition}
%\paragraph{REMARKS} 
\begin{remark}
Notice that $\coeff$ can be replaced by $\samplecoeff$ when we have a sample of $n-1$ dimension.
\end{remark}

\begin{theorem} [\cite{10.1007/978-3-7091-9459-1_12}, Theorem 1]\label{th:mc}
Let $F$ be a finite squarefree basis in $\Z[\bar{x}]$ where $\bar{x}=(x_1,\ldots,x_n)$ and $n\geq 2$ and $S$ be a connected submanifold of     $\R^{n-1}$ such that
each element of $\proj_{mc}(F,x_n)$  is order-invariant in $S$. Then each element in $F$ either
vanishes identically on $S$ or  is analytic delineable on $S$, the sections over $S$ of the elements of $F$ which do not vanish identically on $S$ are pairwise disjoint, and each element of $F$  which does  not vanish identically on $S$ is order-invariant in every such section.
\end{theorem}

%Finally, we define the cell to which sample belonged 
Now, let us use the following definition to describe the procedure of calculating sample cells. We denote by $\factor(A)$ the set of irreducible factors of all polynomials in $A$.
\begin{definition}
Suppose $a=(a_1,\ldots,a_{n-1})$ is a sample of $(x_1,\ldots,x_{n-1})$ in $\R^{n-1}$ and $F=\{f_1,\ldots , f_r\}$ is a polynomial set in $\Z[\bar{x}]$ where $\bar{x}=(x_1,\ldots,x_{n})$. The {\em sample cell} of $F$ at $a$ is
	$$\samplecell(F,a)=\sampleinterval(F_1,\alpha_1)\land\cdots\land\sampleinterval(F_{n-1},\alpha_{n-1}) $$
    where $\alpha_{n-1}=a$, $F_{n-1}=\factor(\proj_{mc}(\factor(F)))$, $\alpha_i=(a_1,\ldots,a_i)$, and  $F_i=\factor(\proj_{sc}(F_{i+1},x_{i+1},\alpha_{i+1}))$ for $i=1,\ldots,{n-2}$.
\end{definition}
%\paragraph{REMARKS}
\begin{remark} \quad
\begin{itemize}
\item It is a standard way to use $\factor$ to ensure that every $F_i$ is a finite squarefree basis. %, and it is even beneficial to the calculation and will not be a burden.
\item Notice that the complexity of computing sample cell  $\samplecell$ depends on $\sum_{i=1}^{n-1}|F_i|$ where $|F_i|$ means the number of polynomials in $F_i$. %And the upper bound complexity for computing $h_i$ can be derived 
From the recursive relationship $|F_{n-1}|=O(r^2+rn)$, $|F_i|<(3+i+1)|F_{i+1}|,i=1,\ldots,n-2$, it is not hard to know that the complexity of computing $\samplecell$ is $O((r^2+rn)(2+n)^{n-1})$.
%\item When calculating $F_1,\ldots,F_{n-2}$, remember to pay attention that the projected polynomial may vanish identically as described in remark \ref{rm:sc}.
\end{itemize}
\end{remark}

%The following corollary is a description of the correctness of the results.
\begin{corollary}
Let $F=\{f_1(\bar{x},x_n),\ldots,f_r(\bar{x},x_n)\}$ be a polynomial set and $a\in \R^{n-1}$, where $\bar{x}=(x_1,\ldots,x_{n-1})$. If $$\forall b\in \R \;\bigvee_{i=1}^{r} f_i(a,b)\rhd_i 0,$$ where $\rhd_i\in \{>,\geq,=\}$,
then $$\forall \alpha \in \samplecell(\{f_1,\ldots,f_r\},a)\forall b \in \R \; \bigvee_{i=1}^{r} f_i(\alpha,b)\rhd_i 0.$$ 
\end{corollary}
\begin{proof}
It is a direct corollary of Theorem \ref{th:mc} and Theorem \ref{th:sc}. 
\end{proof}

\begin{example}
Suppose $f=ax^2+bx+c$ and $\alpha=(1,1,1)$ is a sample of $(a,b,c)$. Then
\begin{align*}
  &F_3=\factor(\proj_{mc}(\{f\},x))=\factor(\{b^2-4ac,a\})=\{b^2-4ac,a\},\\
  &F_2=\factor(\proj_{sc}(\{b^2-4ac,a\},c))=\factor(\{1,a,-4a\})=\{a\},\\
  &F_1=\factor(\proj_{sc}(\{a\},b))=\{a\}.
\end{align*}
So
$$\samplecell(\{f\},a)=c>\Root(b^2-4au,1)\land a>\Root(u,1),$$
and after simplification
$$\samplecell(\{f\},a)=c>\frac{b^2}{4a}\land a>0.$$
\end{example}

\section{CDCL-style search framework}\label{sec:cdcl}
In this section, we introduce a search framework combined with the new projection operator proposed in the previous section. The main notation and concepts about the search framework are taken from Section 3 of \cite{DBLP:conf/cade/JovanovicM12} and Section 26.4.4 of \cite{DBLP:series/faia/2009-185}.
%We adopt the same search framework as \cite{DBLP:conf/cade/JovanovicM12}.

Let $\bar{x}=(x_1,\ldots,x_n)$ and $\level(x_i)=i$. For a polynomial $f$, a literal $l$ and a clause $c$, we define  $\level(f)=\max(\{\level(a)|a\in \var(f))\}$,  $\level(l)=\level(\poly(l))$ and $\level(c)=\max(\{\level(l)|l \in c\})$. We describe the search framework by transition relations between search states as in \cite{DBLP:conf/cade/JovanovicM12}.

The {\em search states} are indexed pairs of the form $M \| \zeta$, where $\zeta$ is  a finite set of polynomial clauses and $M$  is a sequence of  literals and variable assignments. Every literal is  marked as a decision or a propagation literal.  We denote a  {\em propagation literal}  $l$ by  $c\rightarrow l$ if $l$ is propagated from $c$ and denote  a  {\em decision literal}  $l$ by $l^\bullet$. We denote by $x_i\mapsto a_i$ a  variable assignment. Let $\level(x_i\mapsto a_i)=\level(x_i)$ and %Naturally, we can  construct a assignment 
$v[M]=\{x_i\mapsto a_i|(x_i\mapsto a_i)\in M\}$.
For a set $L$ of literals, $v[M](L)$ means the resulting set of $L$ after applying the assignments of $v[M]$.

Next, we introduce transition relations between search states. Transition relations are specified by a set of transition rules. In the following, we use simple juxtaposition to denote the concatenation of sequences ({\it e.g.}, $M,M'$). We treat a literal or a variable assignment as one-element sequence and denote the empty sequence as $\emptyset$. We say the sequence $M$ is ordered when the sequence  is of the form $$M=[N_1,x_1\mapsto a_1,\ldots,N_{k-1},x_{k-1}\mapsto a_{k-1},N_k] $$
where $N_j$ is a sequence of literals and each literal $l\in N_j$ satisfies $\level(l)=j$. Notice that $N_j$ might be $\emptyset$. We define $\level(M)=k$ even if $N_k=\emptyset$. We use $\sample(M)$ to denote the sample $(a_1,\ldots,a_{k-1})$ of $(x_1,\ldots,x_{k-1})$ in $M$ and $\feasible(M)$ to denote the feasible set of $v[M](N_k)$.  For a new literal $l$ with $x_k\in \var(l)$, we say $l$ is consistent with $M$ if $\feasible([M,l])\neq\emptyset$. 
If $l$ is not consistent with $M$, we define $\conflictcore(l,M)$ to be a minimal set of literals $L$ in $M$ such that $v[M](L\cup\{l\})$ does not have a solution for $x_k$.
%and  $\explain(l,M)=\samplecell(\conflictcore(\lnot l,M)\cup\{l\},\sample(M))\land \conflictcore(\lnot l,M)$.
\begin{remark}
Since there is  only one unassigned variable $x_k$ in  the polynomials in $N_k$, so $\feasible(M)$ can be easily calculated by real-root isolation.
\end{remark}
\begin{definition}
Suppose $l$ is a literal and $M$ is an ordered sequence which satisfies $\level(M)=\level(l)$ and $\lnot l$ is not consistent with $M$. Define the {\em explain clause} of $l$ with $M$ as  
%\begin{align*}
%$$\explain(l,M)=\lnot \samplecell(F,\sample(M))\lor \bigvee_{l'\in\conflictcore(\lnot l,M)}\lnot l'\lor l,$$
$$\explain(l,M)=\lnot(\samplecell(F,\sample(M))\land \conflictcore(\lnot l,M))\lor l,$$
%\end{align*}
where $F=\{\poly(l')|l'\in \conflictcore(\lnot l,M)\}\cup \{\poly(l)\}$.
\end{definition}
Meanwhile, we define the {\em state value} of a literal $l$ as

\[
    \Value(l,M)=\left\{\begin{array}{lcl}
        v[M](l)     &    &\level(l)<k, \\
        \text{{\tt True}}  &   &l\in M, \\
        \text{{\tt False}}  &  &\lnot l\in M, \\
        \text{{\tt undef}}   & &\text{otherwise}.
    \end{array}
    \right.
\]
And for a clause $c$,
\[
    \Value(c,M)=\left\{\begin{array}{lcl}
        \text{{\tt True}} &    &\exists l\in c (\Value(l,M)=\text{{\tt True}}), \\
        \text{{\tt False}} &   &\forall l\in c (\Value(l,M)=\text{{\tt False}}), \\
        \text{{\tt undef}}  &  &\text{otherwise}.
    \end{array}
    \right.
\]
Specially, $\Value(\emptyset,M)=\text{{\tt False}}$.

\begin{definition} A set of rules for transition relations between search states are defined as follows where $c$ is a clause and $l$ is a literal.
\begin{description}
%   	\item[Select-Clause]
%     $$ M\|\zeta\Longrightarrow M\|\zeta\vdash c$$
%      if$\level(c)=\level(M)$ and $\Value(c,M)=\text{{\tt undef}}$.
    \item[Decide-Literal]
  		$$M\|\zeta,c\Longrightarrow M,l^\bullet\|\zeta,c$$
        if $l,l'\in c$, $\Value(l,M)=\Value(l',M)=\text{{\tt undef}}$, $\level(c)=\level(M)$  and $l$ is consistent with $M$.
   \item[Boolean-Propagation]
   	$$M\|\zeta, c\lor l\Longrightarrow M,c\lor l\rightarrow l\|\zeta,c\lor l$$
	if $\Value(c,M)=\text{{\tt False}},\Value(l,M)=\text{{\tt undef}}$, $\level(c\lor l)=\level(M)$  and  $l$ is consistent with $M$.
  \item[Lemma-Propagation]
    	$$M\|\zeta\Longrightarrow M,\explain(l,M)\rightarrow l\|\zeta$$
	if $l\in \zeta$ or $\lnot l \in \zeta$, $\Value(l,M)=\text{{\tt undef}}$, $\level(l)=\level(M)$ and  $\lnot l$ is not consistent with $M$.
\item[Up-Level]
    $$ M\|\zeta\Longrightarrow M,x\mapsto a\|\zeta$$
     if $\forall c\in \zeta\;(\level(c)\neq\level(M)\lor \Value(c,M)=\text{{\tt True}})$, $\level(x)=\level(M)$ and $a\in \feasible(M)$.
         \item[Sat]
    $$ M\|\zeta\Longrightarrow(\text{sat},v[M])$$
    if $\level(M)>n$.
       	\item[Conflict]
    $$ M\|\zeta\Longrightarrow M\|\zeta\not\vdash c$$
     if $\level(c)=\level(M)$ and $\Value(c,M)=\text{{\tt False}}$.
     \item[backtrack-Propagation]
    $$M,E\rightarrow l\|\zeta \not\vdash c\Longrightarrow M\|\zeta\not\vdash R$$
    if $\lnot l\in c,\Value(c,[M,E])=\text{{\tt False}}$ and $R=\resolve(c,E,l)$\footnote{$\resolve(c_1\lor l,c_2\lor \lnot l,l)=c_1\lor c_2.$}. 
    \item[backtrack-Decision]
    	$$M,l^\bullet\| \zeta\not\vdash c \Longrightarrow M\|\zeta,c$$
    if $\lnot l\in c$.
    \item[Skip]
    \begin{align*}
    M,l^\bullet\| \zeta\not\vdash c &\Longrightarrow M\|\zeta\not\vdash c\\
    M,E\rightarrow l\| \zeta\not\vdash c &\Longrightarrow M\|\zeta\not\vdash c
    \end{align*}
    if $\lnot l\not\in c$ .
    
  \item[Down-Level]
  \[
  \begin{array}{ll}
  	M,x\mapsto a\|\zeta\not\vdash c \Longrightarrow M\|\zeta \not \vdash c,&\text{ if }\Value(c,M)=\text{{\tt False}}, \\
    M,x\mapsto a\|\zeta\not\vdash c \Longrightarrow M\|\zeta,c,&\text{ if }\Value(c,M)=\text{{\tt undef}}.
  \end{array}
   \]
   \item[Unsat]
   $$M\|\zeta \not \vdash c \Longrightarrow \text{unsat}$$
   if $\Value(c,M)=\text{{\tt False}}$ and no  assignment or  decide literal in $M$.
   \item[Forget]
     $$ M\|\zeta,c\Longrightarrow M\|\zeta$$
     if $c$ is a learnt clause.
\end{description}
\end{definition}
\begin{remark}
Note that in this framework we rely on the rule {\em lemma-propagation} to guide the search away from conflicting states. %by assignments and polynomial literals. 
When applying lemma-propagation, the most important thing is the explain clause. We cannot simply use the conflicting core as the explain clause, as this will cause explain to be an incorrect lemma because it ignores assignments. Using full CAD to calculate explain is also costly. Thanks to the  sample cell calculated by the novel sample-cell projection operator, we can now efficiently calculate an effective explain to achieve our purpose.
\end{remark}
\begin{theorem}
	Given a polynomial formula $\zeta$ with finitely many clauses, any transition  starting from the initial state $\emptyset\|\zeta$ will terminate either in a state $(sat,v)$, where the assignment $v$ satisfies the formula $\zeta$, or in the $unsat$ state. In the later case, $\zeta$ is unsatisfiable in $\R$.
\end{theorem}
\begin{proof}
By Theorem 1 in \cite{DBLP:conf/cade/JovanovicM12}, if there is a finite set such that all the literals returned every time by calling $\explain$ are always contained in the set, then the above theorem holds. On the other hand, it is not hard to see that all literals that may be generated by $\samplecell$ are determined by finitely many polynomials and their real roots and thus finite. That completes the proof. %Therefore, so  $\explain$ is finite.
\hfill $\blacksquare$ 
\end{proof}
\section{Experiments}\label{sec:exp}
In order to better demonstrate the effectiveness of our algorithm, we have implemented a prototype solver LiMbS\footnote{https://github.com/lihaokun/LiMbS}  which is base on Mathematica 12. The solver is  a clean translation of the algorithm  in this paper. Our solver is compared to the following solvers that have been popular in SMT nonlinear competition: Z3 (4.8.7-1), CVC4 (1.6-2), Yices (2.6.1) and MathSAT5 (5.6.0).

All tests were conducted on 6-Core Intel Core i7-8750H@2.20GHz with 32GB of memory and ARCH LINUX SYSTEM (5.5.4-arch1-1). The timeout is set to be 5 hours. 
 
The examples listed below, which we collect from several related papers, are either special or cannot be well-solved by existing SMT solvers. All  results are listed in Table \ref{tb:exp}.

\begin{example}({\bf Han}\_$n$)\cite{Dai_Han_Hong_Xia_2015} Decide whether 
$$\exists x_1,\ldots,\exists x_n \;(\sum_{i=1}^nx_i^2)^2-4(\sum_{i=1}^nx_i^2x_{i+1}^2)<0$$
where $x_{n+1} = x_{1}$. % The existing SMT solvers fail to solve the example where $n>4$ in given time, But  LiMbS can solve han\_5 in 4 seconds.  
\end{example}
\begin{example}$({\mathbf P})$\\ %we find this  formula which is of large scale and easy to satisfy.\\
$\exists a,\exists b,\exists c,\exists d,\exists e,\exists f (a^2 b^2 e^2+a^2 b^2 f^2+a^2 b^2-a^2 b c d e-a^2 b d f+a^2 c^4 f^4+2 a^2 c^4 f^2+a^2 c^4-3 a^2 c^3 e f^3-3 a^2 c^3 e f+3 a^2 c^2 e^2 f^2+a^2 c^2 e^2+a^2 c^2 f^2+a^2 c^2-a^2 c e^3 f-a^2 c e f-a b^2 d e-2 a b c^3 f^4-4 a b c^3 f^2-2 a b c^3+4 a b c^2 e f^3+4 a b c^2 e f+a b c d^2-2 a b c e^2 f^2+a b c f^2+a b c-a b e f+2 a c^3 d f^3+2 a c^3 d f-4 a c^2 d e f^2-2 a c^2 d e+2 a c d e^2 f+b^2 c^2 f^4+2 b^2 c^2 f^2+b^2 c^2-b^2 c e f^3-b^2 c e f-2 b c^2 d f^3-2 b c^2 d f+2 b c d e f^2+c^2 d^2 f^2+c^2 d^2+c^2 f^2+c^2-c d^2 e f-c e f<0)$
\end{example}
\begin{example}\cite{Hong91comparisonof}\quad
%This is the example hong and  a satisfiable  transformation.
\begin{description}
\item[Hong\_$n$]
$$\exists x_1,\ldots,\exists x_n\;\sum_{i=1}^nx_i^2<1\land\prod_{i=1}^n x_i>1$$
\item[Hong2\_$n$]

$$\exists x_1,\ldots,\exists x_n\;\sum_{i=1}^nx_i^2<2n\land\prod_{i=1}^n x_i>1$$
\end{description}
%CVC4 still can solve this example efficiently.
\end{example}

\begin{example}({\bf C}\_$n$\_$r$)
Whether the distance between the % $n$-dimensional 
ball $B_r(\bar{x})$ and the complement of $B_8(\bar{x})$ is less than $\frac{1}{1000}$?
%we measure the distance between by two different norm: 2-norm and $\infty$-norm, which translate the problem into two variations
% \begin{description}
% \item[C1\_r]
% $$\exists x_1,x_2,x_3,y_1,y_2,y_3 \;\sum_{i=1}^{3}x_i^2<r\land\sum_{i=1}^3y_i^2>8^2\land \sum_{i=1}^3(x_i-y_i)^2<\frac{1}{1000}^2$$
% \item[C2\_r]
% $$\exists x_1,x_2,x_3,y_1,y_2,y_3 \;\sum_{i=1}^{3}x_i^2<r\land\sum_{i=1}^3y_i^2>8^2\land \sum_{i=1}^3|x_i-y_i|<\frac{1}{1000}$$
% \end{description}
$$\exists_{i=1}^{n} x_i,\exists_{i=1}^{n} y_i \;\sum_{i=1}^nx_i^2<r\land\sum_{i=1}^ny_i^2>8^2\land \sum_{i=1}^n(x_i-y_i)^2<\frac{1}{1000^2}$$
%The performance of existing SMT solvers on C1\_r and C2\_r vary greatly, but LiMbS  solves the problem stablely.
\end{example}

\begin{table}[htbp]

    \centering
      \begin{tabular}{|c|c|c|c|c|c|c|}
      \hline
        &ans& LiMbS& Z3&CVC4&MathSAT5&Yices\\
      \hline
      Han\_3&SAT& 0.01s&0.01s&0.01s&0.01s&0.01s\\
      \hline
      Han\_4&UNSAT &0.08s&0.01s&$>5$h&$>5$h&0.01s\\
      \hline
      Han\_5&UNSAT& 1.26s&$>5$h&$>5$h&$>5$h&$>5$h\\
      \hline
      Han\_6&UNSAT &60s&$>5$h&$>5$h&$>5$h&$>5$h\\
      \hline
      P   &SAT&1.06s&0.05s&$>5$h&$>5$h&$>5$h\\
      \hline
      Hong\_10 &UNSAT&222s&2058s&0.01s&0.10s&$>5$h\\
      \hline
      Hong\_11 &UNSAT&806s&6357s&0.01s&0.10s&$>5$h\\
      \hline
      Hong2\_11 &SAT&30.43s&1997s&0.01s&$>5$h&0.01s\\
      \hline
      Hong2\_12 &SAT&563s&6693s&0.01s&$>5$h&0.01s\\
      \hline
      C\_3\_1& UNSAT & 0.44s&$>5$h&0.62s&5811s&$>5$h\\
      \hline
      C\_3\_32& UNSAT & 0.48s&$>5$h&unknown&$>5$h&$>5$h\\
      \hline
      C\_3\_63& UNSAT &0.48s&$>5$h&unknown&$>5$h&$>5$h\\
      \hline
      C\_3\_64& SAT & 0.02s&4682s&unknown&$>5$h&$>5$h\\
      \hline
      % C2\_1& UNSAT & 1.80s&0.62s&0.25s&0.27s&0.36s\\
      % \hline
      % C2\_32& UNSAT & 1.77s&0.66s&0.36s&70.29s&$>5$h\\
      % \hline
      % C2\_63& UNSAT &1.82s&0.73s&0.58s&$>5$h&0.32s\\
      % \hline
      % C2\_64& SAT & 0.02s&0.02s&unknown&$>5$h&0.01s\\
      % \hline
      C\_4\_1& UNSAT &1.31s &$>5$h&2.28s&$>5$h&$>5$h\\
      \hline
      C\_4\_32& UNSAT &1.42s &$>5$h&unknown&$>5$h&$>5$h\\
      \hline
      C\_4\_63& UNSAT &1.42s &$>5$h&unknown&$>5$h&$>5$h\\
      \hline
      C\_4\_64& SAT &0.02s &$>5$h&unknown&$>5$h&$>5$h\\
      \hline
      C\_5\_1& UNSAT &5.48s &$>5$h& 19.33s&$>5$h&$>5$h\\
      \hline
      C\_5\_32& UNSAT &5.73s &$>5$h&unknown&$>5$h&$>5$h\\
      \hline
      C\_5\_63& UNSAT &5.68s &$>5$h&unknown&$>5$h&$>5$h\\
      \hline
      C\_5\_64& SAT &0.02s &$>5$h&unknown&$>5$h&1.75s\\
      \hline
      \end{tabular}
      
  \caption{Comparison with other solvers on 21 examples\label{tb:exp}}
  
  \end{table}

Our solver LiMbs solves all the $21$ examples shown in Table \ref{tb:exp}. LiMbs is faster than the other solvers on 15 examples. Only LiMbs can solve 9 of the examples within a reasonable time while other solvers either run time out or return unknown state. 
From this we can see that our algorithm has great potential in solving satisfiability of polynomial formulas, especially considering that our prototype solver is a small program with less than 1000 lines of codes.
For Hong$\_n$ and Hong2$\_n$, though our solver is much faster than Z3, CVC4 is the one that performs best. We note that the examples of Hong$\_n$ and Hong2$\_n$ are all symmetric. This reminds us it is  worth exploiting symmetry to optimize our solver's performance.

\section{Conclusions}\label{sec:conclusion}
A new algorithm for deciding the satisfiability of polynomial formulas over the reals is proposed.  The key point is that we design a new projection operator, the sample-cell projection operator, which can efficiently guide CDCL-style search away from conflicting states. Preliminary evaluation of the prototype solver LiMbS shows the effectiveness of the new algorithm.

We will further develop our algorithm, looking into problems with symmetry, equations or other special structures. %and finding possible optimization points. 
We also hope to develop an easy-to-use, robust and concise open-source algorithm framework based on our prototype solver to achieve a wider range of applications.

\section*{Acknowledgement} This work was supported partly by NSFC under grants 61732001 and 61532019.
\bibliographystyle{splncs04}
\bibliography{samplecad}

\begin{thebibliography}{10}
\providecommand{\url}[1]{\texttt{#1}}
\providecommand{\urlprefix}{URL }
\providecommand{\doi}[1]{https://doi.org/#1}

\bibitem{BCD+11}
Barrett, C., Conway, C.L., Deters, M., Hadarean, L., Jovanovi{'{c}}, D., King,
  T., Reynolds, A., Tinelli, C.: {CVC4}. In: Proceedings of the 23rd
  International Conference on Computer Aided Verification (CAV '11). pp.
  171--177 (2011)

\bibitem{DBLP:series/faia/2009-185}
Biere, A., Heule, M., van Maaren, H., Walsh, T. (eds.): Handbook of
  Satisfiability, Frontiers in Artificial Intelligence and Applications,
  vol.~185. {IOS} Press (2009)

\bibitem{DBLP:journals/corr/abs-1905-09227}
Brau{\ss}e, F., Korovin, K., Korovina, M.V., M{\"{u}}ller, N.T.: A cdcl-style
  calculus for solving non-linear constraints (2019),
  \url{http://arxiv.org/abs/1905.09227}

\bibitem{brown_improved_2001}
Brown, C.W.: Improved {Projection} for {Cylindrical} {Algebraic}
  {Decomposition}. Journal of Symbolic Computation  \textbf{32}(5),  447--465
  (2001)

\bibitem{mathsat5}
Cimatti, A., Griggio, A., Schaafsma, B., Sebastiani, R.: {The MathSAT5 SMT
  Solver}. In: Piterman, N., Smolka, S. (eds.) Proceedings of TACAS. LNCS,
  vol.~7795, pp. 93--107. Springer (2013)

\bibitem{DBLP:journals/cca/Collins76}
Collins, G.E.: Quantifier elimination for real closed fields by cylindrical
  algebraic decomposition. In: Automata Theory and Formal Languages, LNCS 33.
  pp. 134--183 (1975)

\bibitem{Dutertre:cav2014}
Dutertre, B.: Yices 2.2. In: Biere, A., Bloem, R. (eds.) Computer-Aided
  Verification (CAV'2014). LNCS, vol.~8559, pp. 737--744. Springer (2014)

\bibitem{DBLP:conf/cav/DutertreM06}
Dutertre, B., de~Moura, L.M.: A fast linear-arithmetic solver for {DPLL(T)}.
  In: Proceedings of {CAV} 2006. pp. 81--94 (2006)

\bibitem{Dai_Han_Hong_Xia_2015}
Han, J., Dai, L., Hong, H., Xia, B.: Open weak cad and its applications.
  Journal of Symbolic Computation  \textbf{80},  785--816 (2017)

\bibitem{Han_Dai_Xia_2014}
Han, J., Dai, L., Xia, B.: Constructing fewer open cells by gcd computation in
  cad projection. In: Proceedings of ISSAC’14. pp. 240--247 (2014)

\bibitem{DBLP:conf/issac/Hong90}
Hong, H.: An improvement of the projection operator in cylindrical algebraic
  decomposition. In: Proceedings of {ISSAC} '90, 1990. pp. 261--264 (1990)

\bibitem{Hong91comparisonof}
Hong, H.: Comparison of several decision algorithms for the existential theory
  of the reals. Tech. rep. (1991)

\bibitem{DBLP:conf/cade/JovanovicM12}
Jovanovic, D., de~Moura, L.M.: Solving non-linear arithmetic. In: Automated
  Reasoning - 6th International Joint Conference, {IJCAR} 2012. pp. 339--354
  (2012)

\bibitem{DBLP:conf/smt/KorovinKS14}
Korovin, K., Kosta, M., Sturm, T.: Towards conflict-driven learning for virtual
  substitution. In: Proceedings of {SMT} 2014. p.~71 (2014)

\bibitem{DBLP:conf/cav/LiangRTBD14}
Liang, T., Reynolds, A., Tinelli, C., Barrett, C.W., Deters, M.: A {DPLL(T)}
  theory solver for a theory of strings and regular expressions. In:
  Proceedings of {CAV} 2014. pp. 646--662 (2014)

\bibitem{DBLP:journals/jsc/McCallum88}
McCallum, S.: An improved projection operation for cylindrical algebraic
  decomposition of three-dimensional space. J. Symb. Comput.  \textbf{5}(1/2),
  141--161 (1988)

\bibitem{10.1007/978-3-7091-9459-1_12}
McCallum, S.: An improved projection operation for cylindrical algebraic
  decomposition. In: Caviness, B.F., Johnson, J.R. (eds.) Quantifier
  Elimination and Cylindrical Algebraic Decomposition. pp. 242--268. Springer
  Vienna, Vienna (1998)

\bibitem{deMoura+Dutertre+Shankar:cav2007}
de~Moura, L., Dutertre, B., Shankar, N.: A tutorial on satisfiability modulo
  theories. In: Proceedings of CAV'2007. LNCS, vol.~4590, pp. 20--36 (2007)

\bibitem{DBLP:conf/tacas/MouraB08}
de~Moura, L.M., Bj{\o}rner, N.: {Z3:} an efficient {SMT} solver. In: Tools and
  Algorithms for the Construction and Analysis of Systems, 14th International
  Conference, {TACAS} 2008. pp. 337--340 (2008)

\bibitem{DBLP:journals/cacm/MouraB11}
de~Moura, L.M., Bj{\o}rner, N.: Satisfiability modulo theories: introduction
  and applications. Commun. {ACM}  \textbf{54}(9),  69--77 (2011)

\bibitem{10.1007/978-3-7091-9459-1_3}
Tarski, A.: A Decision Method for Elementary Algebra and Geometry. University
  of California Press (1951)

\bibitem{Xia_Yang_2016}
Xia, B., Yang, L.: Automated Inequality Proving and Discovering. WORLD
  SCIENTIFIC (Aug 2016)

\end{thebibliography}
\end{document}